\documentclass[10pt]{article}
\makeatletter
\def\ps@headings{%
\def\@oddhead{\mbox{}\scriptsize\rightmark \hfil \thepage}%
\def\@evenhead{\scriptsize\thepage \hfil \leftmark\mbox{}}%
\def\@oddfoot{}%
\def\@evenfoot{}}
\makeatother
\usepackage{graphicx,amsmath,amssymb,latexsym,epsfig,epstopdf, cite}
\usepackage{subfigure}
\usepackage{flushend}

\usepackage{amsthm}

\newtheorem{proposition}{Proposition}

\def\etal{{\em et al. }}

\title{Getting routers out of the core: Building an optical wide area network with ``multipaths''}

\author{Davide Cuda$^1$, Raluca-Maria Indre$^1$,\\
Esther le Rouzic$^1$,
James Roberts$^2$ \\
{\small \{davide.cuda,ralucamaria.indre,esther.lerouzic\}@orange.com,} \\
{\small james.roberts@inria.fr }\\
$^1$  Orange Labs, France \hspace{10mm}$^2$ INRIA, France}

\begin{document}

\maketitle

\begin{abstract}

We propose an all-optical networking solution for a wide area network (WAN) based on the notion of multipoint-to-multipoint lightpaths that, for short,  we call ``multipaths''. A multipath concentrates the traffic of a group of source nodes on a wavelength channel using an adapted MAC protocol and multicasts this traffic to a group of destination nodes that extract their own data from the confluent stream. The proposed network can be built using existing components and appears less complex and more efficient in terms of energy consumption than alternatives like OPS and OBS. The paper presents the multipath architecture and compares its energy consumption to that of a classical router-based ISP network. A flow-aware dynamic bandwidth allocation algorithm is proposed and shown to have excellent performance in terms of throughput and delay.
\end{abstract}

\section{Introduction}
\label{sec:introduction}
 As Internet traffic continues to grow at a fast pace, there is increasing concern that the current networking paradigm based on large electronic routers will not scale. Scalability in capacity is mainly limited by the power consumption of the core networking systems \cite{Yoo11}. Note also that the Internet counts for a large and increasing fraction of world energy consumption \cite{SMART}. Moreover, as traffic per user grows faster than the number of users, the contribution of the network core is expected to overtake that of the access by 2015 \cite{Kilper11}. Researchers are increasingly looking towards dynamic optical switching technologies as a solution to this problem \cite{ecoc2009}. 
 
Diverse approaches are being explored including optical packet switching (OPS), optical burst switching (OBS) and optical circuit switching (OCS). Unfortunately, the proposed solutions either remain futuristic for want of adequate photonics technology or are ill-suited to the burstiness and fine granularity of Internet traffic. Moreover, it appears doubtful that OPS in particular is in fact competitive with electronics in terms of energy consumption \cite{Tucker11}. On the other hand, optical technology has already been introduced successfully in the access network in the form of the passive optical network (PON).

A PON dynamically shares an optical channel, or lightpath, using time division multiplexing under the control of a medium access control (MAC) protocol. Signals in the form of bursts between end-users and their access node are combined or distributed  passively by means of an optical splitter. A similar lightpath sharing idea is the basis of time-domain wavelength interleaved networking (TWIN) where signals are passively merged and distributed in wavelength selective optical cross connects (OXC) \cite{WSGM03,Saniee09}. TWIN is suitable for an all-optical metropolitan area network (MAN). The present paper proposes an original architecture that assembles shared lightpaths to form a wide area network (WAN).

We propose to build an optical WAN using multipoint-to-multipoint lightpaths that, by contraction, we call ``multipaths''. Multipaths are created using OXCs that either merge, split or distribute the wavelengths of the interconnected fibres. Sharing of the multipath is controlled using an original MAC protocol, inspired by the polling schemes used in PONs and executed by a so-called multipath controller.  Multipaths are used to interconnect electronic edge routers that either concentrate the traffic of end-users or provide access to neighboring networks and data centres.  We explain how controlled multipath sharing completely removes the need for electronic routers in the core, bringing significant reductions in energy consumption. 

As for PONs, multipath sharing can be controlled using a wide range of MAC protocols, each one designed to fulfill the requirements of a specific type of network. In this paper we propose one particular MAC that we believe is suitable for the WAN of an Internet service provider (ISP). It uses a grant-report polling scheme similar to that of the Ethernet PON (EPON). As for EPON, the MAC offers flexibility in the way reports are formulated and grants are computed. We further propose a particular original flow-aware dynamic bandwidth allocation (DBA) algorithm and demonstrate its excellent performance characteristics using simulation.

In the next section we define the multipath. Section \ref{sec:energy} shows how multipaths can be used to build a large WAN typical of a countrywide ISP network and evaluates the resulting energy savings. The proposed MAC protocol and DBA algorithm are described in Section \ref{sec:dba}. Section \ref{sec:performance}  presents the results of the performance evaluation.
 
\section{The multipath}
\label{sec:multipath}
We first discuss related work on lightpath sharing before presenting the multipath and its control mechanisms. We then explain how this structure can be used to create a wide area network without the need for transit routers.

\subsection{Sharing lightpaths}
\label{sec:sharing}

WDM is currently used to provide high-capacity point-to-point links in the form of {\it lightpaths}.  A lightpath consists of a wavelength channel carried over a succession of fibres interconnected by wavelength selective optical cross connects (OXC). As wavelengths have a capacity of 10 Gb/s or more, a lightpath is only used efficiently when its end points concentrate a large amount of traffic. Traffic is usually concentrated through a hierarchy of electronic transit routers. Lightpath sharing can also be performed in the optical domain by employing OXCs to merge or split signals on the same wavelength. 

Point-to-multipoint sharing of lightpaths is realized using splitters to disseminate the wavelength signal over multiple outgoing fibres. A proposal to use optical time division multiplexing is described by Petracca \etal \cite{PMLN03}. Each destination extracts signals in its dedicated time slots. A more flexible alternative would be to require each destination to convert the composite optical signal to electronic form before sorting the packets received and retaining only its own. 

Multipoint-to-point lightpaths are realized by merging wavelength signals on two or more incoming fibres to a single outgoing fibre. Chlamtac and Gumaste \cite{CG03} and Bouabdallah \cite{Bou07} share an optical bus in this way using centralized and distributed MAC protocols, respectively, designed to avoid collisions at the merge points. A more flexible shared lightpath in the form of a concentration tree is the basis of time-domain wavelength interleaved networking (TWIN) introduced by Widjaja \etal \cite{WSGM03}.

A distributed MAC protocol for TWIN, proposed by Sani\'ee and Widjaja in \cite{Saniee09}, is particularly well-suited for building a metropolitan area network (MAN)  (see also \cite{RR11}). Each destination is the root of a multipoint-to-point lightpath and controls sharing by implementing an adapted MAC protocol. Although generalizations have been envisaged, the basic TWIN architecture is limited to MANs for scalability reasons. In addition, the distributed MAC of \cite{Saniee09} or \cite{RR11} is satisfactory only when exchanges between source and destination are timely, which is not physically possible for all nodes in a WAN. 

Note finally that passive optical networks (PONs) also share lightpath capacity between a set of access nodes \cite{G9843, ieee802.3}. PONs use point-to-multipoint lightpaths downstream and multipoint-to-point lightpaths upstream. A proposal by Roger \etal \cite{RON09} would create a MAN by suitably assembling PONs and coordinating their MAC protocols. Our own proposition described next is similar in principle while extending coverage to the wide area.

\subsection{Multipoint-to-multipoint lightpaths}

Multipaths extend the reach of TWIN by terminating the concentration tree at a root close to the source nodes and adjoining a distribution tree to bring the signal from this root to a set of destination nodes. This structure is illustrated in Figure \ref{fig:multipath}.

\begin{figure}[htp]
\centering
\includegraphics[width=.7\columnwidth]{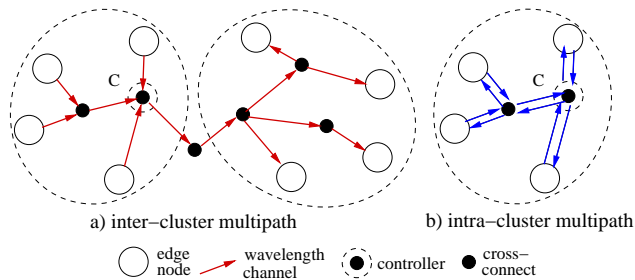}
\caption{Multipoint-to-multipoint lightpaths: multipaths}
\label{fig:multipath}
\end{figure}

In the left hand figure, the multipath interconnects a cluster of source nodes on the left to a cluster of destinations on the right. Access is governed by a controller located at the root of the concentration tree. To ensure timely signalling exchanges, the source nodes are   assumed to be geographically close to their controller (within 100 km, say). Optical signals are transmitted \emph{transparently} from source to destination. Each destination receives a composite signal and extracts its own packets after conversion to electronic form.  Destinations may be a long way from the source cluster and signal amplification and regeneration may be required. 

The right hand figure depicts a multipath that interconnects the nodes of the same cluster. Each node is assumed to be equipped with at least one incoming and one outgoing fibre so that the same wavelength can be used in both directions. This intra-cluster multipath in effect constitutes a MAN.

Optical burst transmissions in inter- and intra-cluster multipaths are coordinated by a controller denoted $C$ in the figure, located at the root of the concentration tree.

\subsection{Multipath control}

\label{sec:coll}

Controller $C$ is responsible for allocating multipath time slots to sources that avoid collisions. It is easy to see that slots cannot collide at any OXC in the multipath if they are timed not to collide at $C$. A MAC protocol that avoids such collisions while using lightpath capacity to the full is described in Section \ref{sec:dba} and evaluated in Section \ref{sec:performance} below. This protocol is inspired by the EPON MAC layer \cite{ieee802.3} that was previously applied also to TWIN \cite{RR11}. The present protocol differs significantly, however, since it is no longer appropriate here to piggyback control signals on the multipath data channel.

Data are transmitted transparently end-to-end in the form of optical bursts. To perform control operations, on the other hand, it is necessary to convert signals to the electronic domain. We therefore constitute dedicated upstream and downstream control channels between controller and source nodes. The upstream channel allows source nodes to communicate their requirements. The downstream control channel is used to inform sources of their  allocated time slots.

The same control channel can be used for all multipaths having the same controller. This motivates the decision to group sources in clusters and to constitute sets of multipaths whose sources are all contained in the same cluster. Multipath destinations can be chosen freely, however. 

The multipaths of the same cluster together with their upstream control channel must use exactly the same fibres. This is necessary for synchronization and ranging operations that are performed over the control channel. It is thus necessary that this channel have the same propagation time as the multipaths it controls (to within the limits implied by chromatic dispersion).

To centrally control a group of multipaths has a further advantage. It is possible to coordinate time slot allocations to avoid transmitter blocking. This phenomenon occurs in TWIN when destinations independently issue grants to a source without regard to the fact that the source may only be equipped with a single tunable transmitter \cite{Saniee09,RR11}. The cluster controller is aware of all assigned grants and can avoid collisions both at the source and on the multipath.

Finally, to use dedicated control channels allows the controller to more closely monitor source activity and notably, to avoid needlessly assigning slots to sources that momentarily have no packets to emit. Such slots are necessary with piggybacked signalling to allow the source to inform the controller in a timely manner of new arrivals to its queues.

\subsection{Building a WAN}

In Section \ref{sec:energy} below we consider how multipaths enable one to remove some 92 electronic core routers from a hypothetical country-wide ISP network. In the present section we aim simply  to demonstrate the scalability of multipaths using a toy network example. 

 \begin{figure}[htp]
 \centering
 \includegraphics[width=.7\columnwidth]{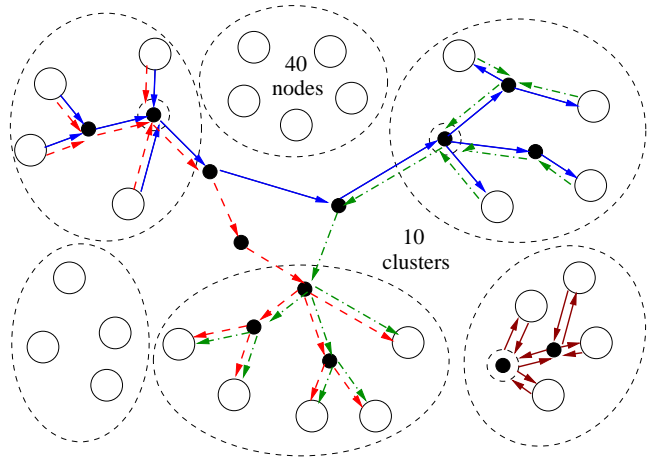}

 \caption{$10$ cluster network showing just 4 multipaths out of $100$}
 \label{fig:multipathnet}
 \end{figure}

We consider a symmetrical network of 400 edge nodes. Each node exchanges 2 Gb/s exclusively with the other nodes. We constitute 10 clusters of 40 nodes each. Both inter- and intra-cluster traffic is less than 8 Gb/s and can be handled by the 10 Gb/s optical channel of a single multipath.  The 400 nodes are thus interconnected using a total of 100 multipaths. Figure \ref{fig:multipathnet} illustrates some of the multipaths of this network.

Each node must be able to receive 10 wavelengths and requires a dedicated receiver for each. It must also be able transmit on 10 wavelengths but a single tunable transmitter is sufficient since total outgoing traffic is only 2 Gb/s.
 Each node must also have access to the control channels of its source cluster.  

In contrast, if the network were composed of point-to-multipoint or multipoint-to-point lightpaths, it would be necessary for each node to receive or transmit, respectively, on a total of 399 distinct wavelengths. It is necessary to terminate a large number of outgoing fibres even though overall demand is well within the capacity of a single channel. Nodes must also be equipped with a large number of transmitters or receivers.

\subsection{Managing a multipath network}

Multipaths can be rapidly set up and torn down in the optical infrastructure using a control plane like GMPLS \cite{GMPLS}. Reliability can thus be assured using familiar routing and wavelength assignment techniques although additional provision would be required to ensure adequate controller redundancy. 

A control plane like GMPLS would makes it relatively easy to reconfigure the multipaths so as to meet changing traffic requirements. The capacity of the network could be reduced in times of light traffic by using bigger clusters. The total number of multipaths would then be reduced and redundant equipment could be switched off. More significant gains in energy consumption result from using multipaths in place of electronic transit routers, as considered next.

\section {Realizing a large ISP multipath network}
\label{sec:energy}

In order to highlight the multipath architecture's potential, we demonstrate how one could take the routers out of the core of a representative ISP network. Classical and multipath networks are compared in terms of their energy consumption.   We then discuss the scope for multipaths to meet anticipated future growth in Internet traffic.

\subsection{An ISP reference network}

The considered network handles some 4 Tb/s of traffic and is intended to be representative of a large national ISP. It interconnects 420 edge nodes and provides access to the Internet via 4 gateways. Additionally, 2 peering points provide interconnection to peers or data centres.

We make the following assumptions on exchanged traffic in the busy period. Each edge node sends a total of 1 Gb/s  to the other edge nodes within the WAN. This is assumed to be distributed uniformly over all other edge nodes. Edge nodes also send 1 Gb/s to the Internet via one of the 4 equivalent gateways. Traffic outgoing via the 2 peering points is assumed relatively very small. Nodes receive 4 Gb/s  from the Internet via one of the 4 gateways and 2 Gb/s of traffic from the 2 peering points combined. 

 \begin{figure}[htp]
 \centering
 \includegraphics[width=.7\columnwidth]{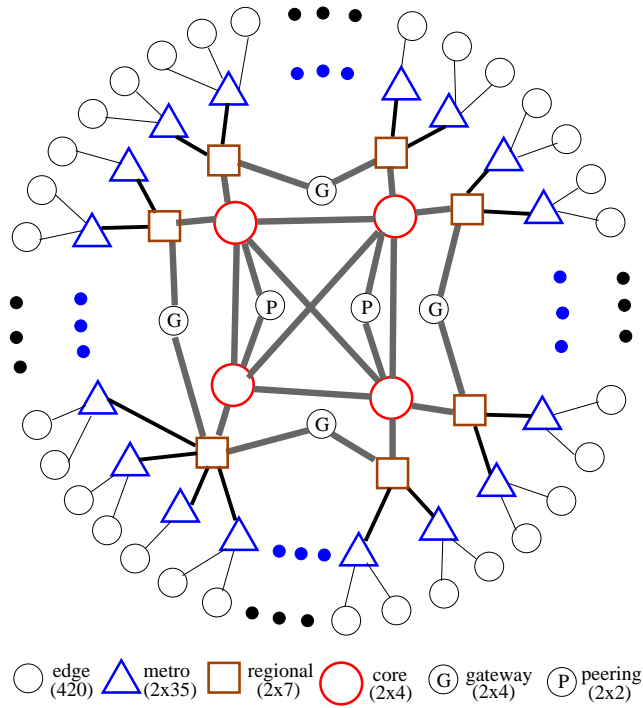}
 \caption{ISP network: all nodes except user edges are duplicated}\label{fig:ISP}
 \end{figure}

 \begin{figure}[htp]
 \centering
 \includegraphics[width=.7\columnwidth]{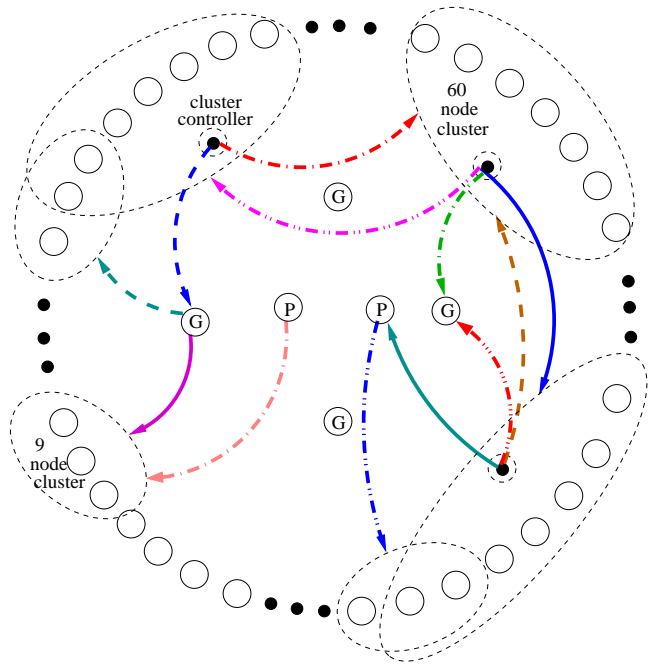}
 \caption{ISP network realized using multipaths showing a sample of paths}
 \label{fig:ISPmulti}
 \end{figure}

The reference network is illustrated in Figure \ref{fig:ISP}. Edge routers are interconnected via three hierarchical layers of high capacity transit routers. Internet gateways are connected at the level of regional transit nodes. Peering is realized via the top level core routers. For reasons of reliability, each transit node, gateway and peering point is duplicated (the figure shows only one member of each pair). The routers are interconnected by point-to-point links as follows: 

\begin{itemize}
\item[$\bullet$] 10 Gb/s from edge node to  metro node,
\item[$\bullet$] 80 Gb/s from metro node regional node, 
\item[$\bullet$] 200 Gb/s from regional node to core node,
\item[$\bullet$] 200 Gb/s from core node to core node,
\item[$\bullet$] 280 Gb/s from gateway to regional node,
\item[$\bullet$] 240 Gb/s from peering point  to core node. 
\end{itemize}

Links are sized to support the total traffic originating from or destined to the edge nodes with redundant dual connections to ensure reliability. Note that a total of 92 routers (8 core, 14 regional and 70 metro routers) are required to provide full connectivity to the 420 edge nodes. 

\subsection{Multipath network}
\label{sec:ispmulti}
Consider now how connectivity can be provided using 10 Gb/s multipaths. We maintain edge nodes, gateways and peering points and remove the three transit levels.  Traffic assumptions are as above. 
Each edge node is connected to the optical infrastructure by one incoming and one outgoing fibre. It is equipped with one 10 Gb/s tunable transmitter and one receiver for each multipath it terminates. We assume the performance of each multipath is satisfactory as long as its load is less than 90\% (see Section \ref{sec:performance}). 

Figure \ref{fig:ISPmulti} represents a sample of multipaths, shown here as arcs. The endpoints of each arc are connected to nodes in the corresponding cluster by a tree network, as previously described, but the branching is omitted for the sake of simplicity. 

For internal traffic, seven clusters are constituted with 60 nodes each and  destination clusters coincide with source clusters. The traffic from one cluster to another is then 8.5 Gb/s and can be handled by one multipath. Intra-cluster multipaths have slightly less traffic than 8.5 Gb/s. 

Incoming Internet traffic is handled by point-to-multipoint multipaths from a gateway or peering point to a cluster of edge nodes. Since each edge node receives 4 Gb/s of traffic from the Internet, clusters of 2 nodes result in a traffic of 8 Gb/s per multipath. Each peering point sends 1 Gb/s of traffic to each node which results in clusters of 9 nodes with a traffic of 9 Gb/s per multipath. 

Multipaths for outgoing Internet traffic connect each 60-node source cluster to a gateway. The 60 Gb/s of outgoing cluster traffic can be handled by 7 such multipaths, distributed over the 4 gateways. The small amount of traffic towards each peering point can be handled by one further multipath from each source cluster. 

As in the ISP network, gateways and peering points are duplicated for reliability. The same multipaths can be used to interconnect clusters with each pair, only one gateway or peering point being active at any time. Overall network reliability is assured at the optical infrastructure level. This is possible since multipaths can be reconfigured rapidly in the event of failure. We do not account explicitly for redundant fibre and cross-connects since this is more or less common to both transit router and multipath architectures and consumes relatively little energy.

Edge nodes only generate 2Gb/s of outgoing traffic so a single tunable transmitter is sufficient for 16 multipaths (6 inter-cluster, 1 intra-cluster, 7 gateway and 2 peering point paths) and the control channel. They must be equipped to receive a total of 9 multipaths (6 inter-cluster, 1 intra-cluster, 1 gateway and 2 peering point paths) and one control channel, i.e., 10 receivers in all. Each gateway or peering point must be able to emit on 53 different channels and to receive up to 13 different multipaths. 

\subsection{Comparing energy consumption}

To estimate the power consumption of the ISP network, we use the model proposed in \cite{Filip}, taking data from vendor-supplied specifications. Table \ref{tab:routers} gives our power consumption estimates for the different network nodes. 
The indicated power consumption includes air conditioning which we consider to contribute 33\% of the device consumption. 

Overall consumption amounts to 6.4 MW of which 1.7 MW is due to transit switching. The latter is economized in the multipath architecture at the cost of some additional optical equipment. The main item is the set of receivers with which edge nodes 
must be equipped, one for each multipath. 

The total number of receivers required in the network of Section \ref{sec:ispmulti} is 4200. Assuming a consumption of 17.5 W per receiver (half the consumption of a transceiver) \cite{transp} yields only 73,5 kW. The energy cost of a controller should be equivalent to that of an EPON OLT that we estimate at 100 W \cite{OLT} and is thus negligible. 

It is not obvious that the multipath architecture would require more optical infrastructure. However, we consider any corresponding increase in consumption to be negligible since that of an OXC is typically not more than 2.5 kW \cite{OXC}. In the multipath architecture, gateways generate and terminate a slightly higher number of lighpaths which increases their consumption by 15\%. 

In conclusion, the multipath architecture reduces consumption in the considered ISP network by some 27\%. The gain is clearly much more significant if we compare only the part that changes: multipaths consume around 100 kW compared to the 1.7 MW energy cost of electronic transit routers.

\begin{table}[ht]
\begin{center}
\caption{\label{tab:routers} Router power consumtion.}
\begin{tabular}{| c | c | c | c | } \hline
Node &  Power  &  Number  &  Total \\
\hline
\hline
core & 47 kW & 8 & 376 kW \\
\hline
regional & 44 kW & 14 & 616 kW  \\
\hline
metro & 10 kW & 70 & 700 kW \\
\hline
edge & 10 kW & 420 & 4200 kW  \\
\hline
gateway & 45 kW & 8 & 360 kW \\
\hline
peering point & 38 kW & 4 & 152 kW  \\
\hline

\end{tabular}
\end{center}
\end{table}

\subsection{A future proof technology}

Consider now the potential of multipaths to cope with growing demand in the Internet. One could increase the capacity of each wavelength channel or  increase the number of channels per fibre. 

Suppose a tenfold increase of the traffic generated by each edge node of the ISP network. The same set of multipaths would be sufficient if each operates at 100 Gb/s instead of 10 Gb/s. Note that 100 Gb/s optical transmission systems are expected to become commercially available in the near future \cite{100G}.  

A more cost-effective approach would be to retain the same 10 Gb/s transmitters and receivers and to create smaller clusters. One might alternatively allocate more than one wavelength channel to the same source-destination cluster pair. Despite the higher number of wavelength channels required, this solution appears perfectly feasible with present day WDM systems which have a count of up to 160 10 Gb/s wavelength channels per fiber \cite{WDM}.

\section {Dynamic allocation of multipath bandwidth}
\label{sec:dba}
The multipath architecture can be used with a wide variety of medium access control protocols designed to satisfy a range of performance requirements. In this section we propose one particular MAC that we believe is suitable for an ISP network like that described above. 

\subsection{Report-grant signalling}
\label{sec:sig}

We propose a report-grant polling scheme inspired by that used for dynamic bandwidth allocation (DBA) in EPON. Multipath source nodes inform the controller of their current requirements (e.g., packet queue lengths) and the controller grants them appropriately sized time slots that avoid collisions. These time slots are specified by their start time and duration. Unlike EPON and the MAN considered in \cite{RR11}, the multipath architecture uses dedicated control channels for reports and grants for all the multipaths of a source cluster.

We assume the control channel implements an optical TDM (OTDM) scheme with fixed slots dedicated to each source-multipath pair. Supposing each report requires 128 B (EPON reports are 64 B), that the number of nodes is 60 and the controller manages 16 multipaths (data from Section \ref{sec:energy}), the signalling requirement is roughly 125 KB per cycle. A control channel rate of 1 Gb/s would thus bring a maximum signalling delay of 1 ms. A downstream channel of the same rate would allow grants to be specified in 128 B which appears largely sufficient (it certainly is for the scheme we propose below).

\subsection{Synchronization and ranging}

Grant timing must take proper account of the different propagation times between controller and source nodes. As the optical multipath channels are transparently switched from source to destination, synchronization must be performed over the control channels whose signals are converted to electronic form at the controller. Since these use exactly the same fibres as the multipaths, their propagation times are identical to within tight limits accounting for chromatic dispersion.

To synchronize clocks and measure propagation times we adapt the time stamp exchange scheme of EPON, as in \cite{RR11}. Each grant signal includes a time stamp $t_1$ corresponding to the time it was sent according to the controller clock. On receipt, the source sets its own clock to the value $t_1$ (it is thus slow with respect to the controller clock by a one-way propagation time). When the source sends a report it includes the time it was sent according to its local clock  $t_2$. The controller receives this report at its local time $t_3$. It may readily be verified that the round trip time is then equal to $t_3 -t_2$. 

The controller thus measures precisely the round trip time $\textsc{rtt}_i$ between itself and each source $i$ of its cluster. Each source maintains a single local clock that is valid for all multipaths managed by the controller. The algorithm for computing grant epochs described next needs only the $\textsc{rtt}_i$ and knowledge that source clocks are slow by precisely the one-way propagation time.

In computing transmission slot starting times it is necessary to account for the random delay taken for the grant to reach the source. This includes the interval between the moment the grant is formulated and the time it can be inserted in the appropriate slot of the control channel OTDM frame. This is bounded by the cycle time (1 ms in the example considered in the previous subsection).

\subsection{Grant allocation}

\label{sec:allocation}
Grant timing must account for a guard time between burst emissions. This is necessary in particular to allow re-tuning of transmitter lasers. We denote the guard time by $\Delta_g$. Note, for instance, that  GPON standards specify a guard time of less than 100 nanoseconds. 

It is also necessary to account for grant signalling delays and for imprecise synchronization (due to chromatic dispersion or coarse clock granularity, for example). We suppose the sum of delay and timing errors are less than some value $\tau$. From previous discussion, a value $\tau =$ 1 ms is not unreasonable.   

The process of grants emitted by the controller to the source nodes of a given multipath $j$ is specified by the functions $g_j(.)$, $s_j(.)$ and $d_j(.)$ defined as follows. The $n^{th}$ grant sent for multipath $j$ to some source is formulated by the controller at time $g_j(n)$ and instructs the source to transmit for duration $d_j(n)$ starting at  {\it source local time} $s_j(n)$. Assume the $(n+1)^{th}$ grant is issued to  source $i$.  Epochs $g_j$ and $s_j$ are calculated recursively, as specified in Proposition \ref{prop:recursion}.

\begin{proposition}
\label{prop:recursion}
The following recursions define a schedule that is feasible and ensures that multipath $j$ is fully utilized, on condition that a transmitter is available at the designated start time:  
 \begin{eqnarray}
  g_j(n+1) &= &g_j(n) + d_j(n) + \Delta_g, \label{eq:g-update}\\
 s_j(n+1) &= &g_j(n+1) + \Delta_O  - \textsc{rtt}_{i}, \label{eq:s-update}
\end{eqnarray}
where $\Delta_O$  is an offset satisfying $\Delta_O \geq  \max_i(\textsc{rtt}_{i})+\tau$.
\end{proposition}

An equivalent proposition is proved in \cite{RR11}. The choice of source $i$ for the $(n+1)^{th}$ grant is not specified in the above recursion. If the choice is made independently for each multipath, it is possible that grants to a given source may overlap and lead to blocking when the number of overlaps is greater than the number of tunable transmitters with which the source is equipped. Such blocking was shown in \cite{RR11} to lead to significant loss of traffic capacity. In the present multipath architecture, blocking can be avoided since all grants to a given source are assigned by the same controller. The controller can select a source $i$ for the  $(n+1)^{th}$ grant that is known to have an available transmitter. This is stated in the following proposition.

\begin{proposition}
\label{prop:coordination}
Assuming the total number of transmitters, summed over all sources, is not smaller than the number of multipaths, the following procedure avoids transmitter blocking. For each $i$, let $\mathsf{Free}_i(t)$ be the time according to its local clock that source $i$ transmitter $t$ becomes free. Let $t_i^* = \arg\min(\mathsf{Free}_i(t))$. The controller attributes the $(n+1)^{th}$ grant for multipath $j$ to some source $i$ for which
$s_j(n+1) \ge {Free}_i(t_i^*)$
and sets
\begin{equation}
\mathsf{Free}_i(t_i^*)= s_j(n+1)+d_j(n+1)+\Delta_g.  \label{eq:end-update}
\end{equation}
\end{proposition}

\begin{proof}
The condition on the number of transmitters ensures that at least one source, $i$ say, has a free transmitter at its local time $s_j(n+1)$ implying that $s_j(n+1) \ge {Free}_i(t_i^*)$. The condition on  $\mathsf{Free}_i(t)$ is sufficient since only the last attributed grant for each $t$ can possibly interfere with the current $(n+1)^{th}$ grant. This is so since the offset between grant time $g$ and start time $s$ for a given source $i$ is the same for all multipaths.
\end{proof}

Note that several sources may satisfy the condition $s_j(n+1) \ge {Free}_i(t_i^*)$.  In the implementation evaluated in Section \ref{sec:performance} we have chosen the first feasible source with positive demand found when testing in cyclic order with random starting point.

\subsection{Flow-aware allocations} \label{subsec:flow}

It is possible to design a large variety of DBA algorithms that exploit the above protocol. In this paper we only consider one such algorithm, the per-flow fair sharing algorithm introduced in \cite{RR11}.

We suppose source nodes reliably identify flows using packet header fields and implement per-flow fair queueing. In fact, we suppose in the evaluations in Section \ref{sec:performance} that the scheduler distinguishes two kinds of flows, backlogged flows and non-backlogged flows and gives priority to packets of the latter. The adaptation of deficit round robin described by Kortebi \etal \cite{KOR05} called PDRR (for priority deficit round robin) realizes this distinction\footnote{In fact, a regular DRR scheduler could be used with broadly similar performance results.}. 

The scheduler automatically determines which flows are backlogged at any time and is also aware of the size in bytes of the priority queue for packets from non-backlogged flows. These data are used to constitute reports. When a report is due, the source communicates the current number of backlogged flows and the size of the priority queue. 

When grants are formulated, the controller accounts for all reports received since the last grant for that source and that multipath and sets the grant size equal to the transmission time of the priority queue content plus one quantum for each backlogged flow. The quantum is typically equivalent to one or several packets. When the grant start time arrives, the source uses it first to empty the priority queue. The residue is used to send quanta for as many backlogged flows as possible, serving these flows in round robin order.

Given the order in which grants from a multipath are assigned to sources (cf. Sec. \ref{sec:allocation}), this flow-aware allocation realizes distributed per-flow fair sharing of multipath bandwidth. This ensures (relatively low rate) streaming and conversational flows suffer negligible loss and delay while the flows for which the multipath is a bottleneck are constrained to a fair bandwidth share.

As in \cite{RR11}, we suppose packets can be fragmented as necessary to completely fill allocated grants. Destination nodes are assumed to convert the optical signal to electronic form and reconstitute the original packets. Each burst must be labelled with the identity of the source that emitted it to enable fragments to be reassembled. The destination recognizes its own packets from the address in the layer 2 or layer 3 header and discards the rest.

\section{DBA performance evaluation}
\label{sec:performance}
We evaluate the performance of a source cluster of $N$ nodes sharing a set of $M$ multipaths. Evaluation is mainly by simulation using a C program that implements the DBA algorithm described in Section \ref{sec:dba}.  

\subsection{Demand and system parameters}
\label{sec:parameters}
To simplify, we suppose flows are either backlogged (the multipath is their throughput bottleneck) or have a sufficiently low rate that their packets are always scheduled with priority. The backlogged flows have an exponentially distributed size of mean 10 MB and arrive according to a Poisson process. The non-backlogged flows have rate 2 Mb/s, arrive according to a Poisson process and have an exponential distribution of mean 30 s. Packets of all flows have constant size 1 KB. These somewhat artificial data are sufficient to illustrate the main performance characteristics of the multipath architecture.

Multipath capacity $C$ is 10 Gb/s. The maximum controller-source round trip time is 1 ms (a distance of 100 km). We assume the maximum grant delay $\tau$ is 1 ms yielding an offset $\Delta_O$ of 2 ms (cf. Sec. \ref{sec:allocation}). The transmitter guard time $\Delta_g$ is 100 ns.  Finally, the quantum size used by the PDRR scheduler is 1 KB, or one backlogged packet.

Traffic demand of source $i$, $1\le i \le N$,  for destinations served by multipath $j$, $1 \le j \le M$, is denoted $a_{ij}$. This is the product of flow arrival rate by mean flow size and is measured in bit/s. The proportion of demand due to backlogged flows is a simulation parameter. Source node $i$ is equipped with $t_i$ tunable transmitters. We generally assume symmetric traffic where all the $a_{ij}$ are equal.

\subsection{Traffic capacity}

Traffic capacity of the considered network is defined as the outer limit of the capacity region, i.e., the set of demands $a_{ij}$ for which source queues remain stable. This is a very important performance measure for high speed optical communications since, as confirmed later, transmission delays are negligibly small until demand approaches this limit.

The following are clearly necessary conditions for stability:
\begin{eqnarray}
\sum_{1\le i \le N} a_{ij}& < &C, \textrm{ for } 1\le j \le M, \label{eq:cap-path} \\ 
\sum_{1\le j \le M} a_{ij}& < &t_i C,  \textrm{ for } 1\le i \le N. \label{eq:cap-ttx} 
\end{eqnarray}

From results proved in \cite{RR11}, we know that an isolated multipath ($M=1$) is stable if and only if condition (\ref{eq:cap-path}) is satisfied. If multipaths are shared with independent, uncoordinated schedules for each, transmitter blocking means traffic capacity is smaller. For example, the analysis in \cite{RR11} shows that up to 37\% of capacity is lost when $t_i=1$. 
The algorithm in Section \ref{sec:allocation} avoids blocking since the controller coordinates grant allocations. 

\subsection{Throughput}

Throughput is defined as the ratio of the mean flow size to the mean flow response time. 
We first investigate the gain in throughput brought by coordinating allocations on different multipaths to avoid transmitter blocking. 

\begin{figure}[h]
\begin{center}
\includegraphics[width=.7\columnwidth]{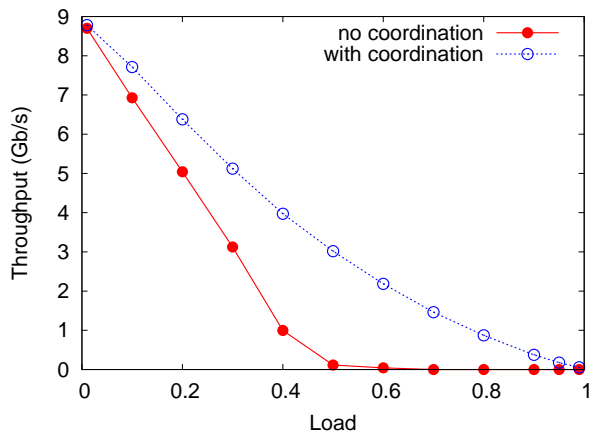}
\caption{\label{fig:comp} Mean throughput of backlogged flows for coordinated and non-coordinated allocations, $N=10$, $M=10$.}
\end{center}
\end{figure}

Figure \ref{fig:comp} shows throughput as a function of multipath load, $\sum_i a_{ij}/C$, for a cluster of 10 sources sharing 10 multipaths with symmetric traffic composed entirely of backlogged flows. Nodes have a single tunable transmitter. The most significant observation is that coordination indeed significantly increases traffic capacity (given by the load at which throughput goes to zero). As predicted by the analytical result in \cite{RR11}, the traffic capacity of a system with independent allocations is as low as 63\% of the raw multipath rate. 

The present DBA performs much better than TWIN at low load because of the use of out-of-band signalling \cite{RR11}. At very low load, with high probability, each flow is alone in the network and the only overhead is one guard time for each quantum of useful traffic.  With the parameters of Section \ref{sec:parameters} (1 quantum transmitted in 800 ns and $\Delta_g=100$ ns), this translates to a maximum throughput of 8.89 Gb/s, with and without coordination.
  
\begin{figure}[h]
\begin{center}
\includegraphics[width=.7\columnwidth]{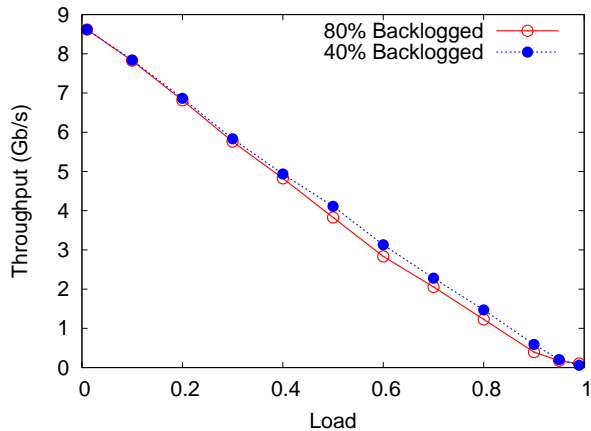}
\caption{\label{fig:thr} Mean throughput of backlogged flows depending on backlogged traffic proportion, $N=60$, $M=16$.}
\end{center}
\end{figure}

Figure \ref{fig:thr} shows that throughput depends significantly on the proportion of backlogged traffic. The network here is  the 60 node cluster with 16 multipaths  envisaged in Section \ref{sec:energy}. Traffic capacity is again maximal and does not depend significantly on the relative proportion of backlogged flow traffic. At medium loads, throughput increases somewhat as the proportion of backlogged traffic decreases.

\subsection{Delay}
The mean delay of packets in the priority queue is plotted against load for the $60\times 16$ network in Figure \ref{fig:delay}. This delay is representative of the per-packet latency experienced by conversational and streaming flows (supposed to have a relatively low intrinsic rate and therefore scheduled with priority). The figure confirms that such latency is minimally impacted by the multipath network. Delay attains 2.5 ms at very low load, decreases slightly as load increases, before increasing sharply as demand approaches capacity. Results for different proportions of backlogged flow traffic are not significantly different.

The delay at low load is composed as follows: a packet arrives; this arrival is signalled in the next report (average time .5 ms); the report arrives after a one-way propagation time (average .25 ms); a grant is formulated immediately and in view of (\ref{eq:s-update}) the sending time is set (2 ms offset -- one-way propagation time = 1.75 ms on average). The total of 2.5 ms could be reduced by implementing a faster signalling channel but not below the maximum round trip time of 1 ms.

Delay decreases with increasing load because the priority packets ``steal'' transmission time previously allocated for backlogged flows. Delay is only significant at loads greater than 95\%, confirming that traffic capacity is the essential performance criterion for high speed optical networks.

\begin{figure}[h]
\begin{center}
\includegraphics[width=.7\columnwidth]{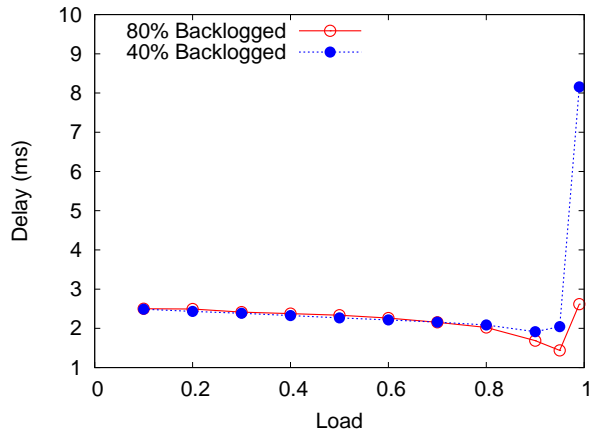}
\caption{\label{fig:delay} Delay performance of the multipath architecture.}
\end{center}
\end{figure}

\section{Conclusion}

\label{sec:conclusion}
The paper proposes a novel approach to realizing an optical wide area internetwork using technology that is available today, namely rapidly tunable transmitters, burst mode receivers and wavelength selective, merge and split capable, cross-connects. This technology is used to create multipoint-to-multipoint lightpaths called multipaths. 
Multipath bandwidth is shared between source-destination traffic flows under the control of a MAC protocol designed to avoid collisions at merge points and destinations. 

Multipaths can be used to interconnect the edge nodes, gateways and peering points of a countrywide ISP network, thereby eliminating the requirement for transit routers. In a considered case study, this would remove 92 routers from a 4~Tb/s network, leading to energy savings  of some 1.7~MW.

The MAC protocol is implemented by a controller common to the set of multipaths used by the edge nodes of a source cluster. It employs an asynchronous polling scheme like that defined in the EPON standard. Unlike EPON, however, multipath report and grant messages are not piggybacked but use dedicated control channels. The MAC layer is flexible in that a range of policies for fixing the size of the grants can be implemented.

We have proposed and evaluated a particular DBA algorithm that realizes cluster-wide per-flow fair sharing of multipath bandwidth. Simulation results demonstrate that, as long as the controller coordinates allocations over multipaths to avoid transmitter blocking,  the network has maximal traffic capacity limited only by the multipath bandwidth.  When demand is less than 95\% of capacity, packet delays for flows that are not backlogged, including conversational and streaming flows, are negligibly small.

The multipath extends the field of application of passive optical technology to large, wide area networks. We believe this approach has considerable potential beyond the context considered here. In future work, we intend to consider its adaptation to large tier-1 and data centre networks.

\bibliographystyle{IEEEtran}
\bibliography{IEEEabrv,mpmp}
\end{document}